\documentclass[a4paper,11pt]{article}
\usepackage{microtype} 
\usepackage{authblk}
\RequirePackage{footmisc}

\usepackage{graphicx}
\usepackage{xspace}
\usepackage{amsmath,amssymb}
\usepackage[caption=false]{subfig}
\usepackage{xcolor}
\usepackage{hyperref}
\usepackage{amsthm,amsfonts,amssymb,amsmath,array}
\usepackage{fullpage}

\usepackage{thmtools,thm-restate}
\usepackage{cleveref}

\newtheorem{theorem}{Theorem}
\newtheorem{lemma}[theorem]{Lemma}

\graphicspath{{./figures/}}

\title{Bounded-Degree Spanners in the Presence of Polygonal Obstacles}

\author[1]{Andr\'e van Renssen\thanks{Email: andre.vanrenssen@sydney.edu.au}}
\author[1]{Gladys Wong\thanks{Email: gwon6099@uni.sydney.edu.au}}

\affil[1]{University of Sydney, Australia}

\date{}

\newcommand{\etal}{\emph{et al.}\xspace}
\newcommand{\graphname}{$G_\infty$\xspace}
\newcommand{\longname}{polygon-constrained half-$\Theta_6$-graph\xspace}

\begin{document}

\maketitle

\begin{abstract}
Let $V$ be a finite set of vertices in the plane and $S$ be a finite set of polygonal obstacles, where the vertices of $S$ are in $V$. We show how to construct a plane $2$-spanner of the visibility graph of $V$ with respect to $S$. As this graph can have unbounded degree, we modify it in three easy-to-follow steps, in order to bound the degree to $7$ at the cost of slightly increasing the spanning ratio to 6. 

\end{abstract}

\section{Introduction}
\label{sec:introduction}
A geometric graph $G$ consists of a finite set of vertices $V\in\mathbb{R}^2$ and a finite set of edges $(p,q)\in E$ such that the endpoints $p, q \in V$. Every edge in $E$ is weighted according to the Euclidean distance, $|pq|$, between its endpoints. For any two vertices $x$ and $y$ in $G$, their distance, $d_G(x,y)$ or $d(x,y)$ if the graph G is clear from the context, is defined as the sum of the Euclidean distance of each constituent edge in the shortest path between $x$ and $y$. A $t$-spanner $H$ of $G$ is a subgraph of $G$ where for all pairs of vertices in $G$, $d_H(x,y) \leq t \cdot d_G(x,y)$. The smallest $t\geq1$ for which this property holds, is called the \emph{stretch factor} or \emph{spanning ratio} of $H$. For a comprehensive overview on spanners, see Bose and Smid's survey \cite{BOSE2013818} and Narasimhan and Smid's book \cite{Narasimhan:2007:GSN:1208237}. Since spanners are subgraphs where all original paths are preserved up to a factor of $t$, these graphs have applications in the context of geometric problems, including motion planning and optimizing network costs and delays. 

Another important factor considered when designing spanners is its maximum degree. If a spanner has a low maximum degree, each node needs to store only a few edges, making the spanner better suited for practical purposes. The best degree bound for plane spanners is 4 by Bonichon~\etal~\cite{Bonichon2015}, whose spanner had a spanning ratio of 156.82. This result was improved by Kanj~\etal~\cite{DBLP:journals/corr/KanjPT16}, who reduced the spanning ratio to 20. Bose~\etal~\cite{DBLP:journals/corr/BiniazBCGMS16} showed that degree 3 can be achieved in two special cases. In terms of lower bounds, Dumitrescu and Ghosh~\cite{dumitrescu2016lower} showed that there exist point sets that require a spanning ratio of at least 1.4308. They also strengthened this bound to 2.1755 for spanners of degree 4 and 2.7321 for spanners of degree 3. 

Most research has focussed on designing spanners of the complete graph. This implicitly assumes that every edge can be used to construct the spanner. However, unfortunately, in many applications this is not the case. In motion planning we need to move around physical obstacles and in network design some connections may not be useable due to an area of high interference between the endpoints that corrupts the messages. This naturally gives rise to the concept of \emph{obstacles} or \emph{constraints}. Spanners have been studied for the case where these obstacles form a plane set of line segments. It was shown that a number of graphs that are spanners without obstacles remain spanners in this setting~\cite{BCR2018GeneralizedDelaunayJournal,basepaper,BK06,BR2019Constrained}. In this paper we consider more complex obstacles, namely simple polygons. 

Let $S$ be a finite set of simple polygonal obstacles where each corner of each obstacle is a vertex in $V$, such that no two obstacles intersect. Throughout this paper, we assume that each vertex is part of at most one polygonal obstacle and occurs at most once along its boundary, i.e., the obstacles are vertex-disjoint simple polygons. Note that $V$ can also contain vertices that do not lie on the corners of the obstacles. Two vertices are \emph{visible} to each other if and only if the line segment connecting them does not properly intersect any obstacles (i.e., the line segment is allowed to touch obstacles at vertices or coincide with its boundary, but it is not allowed to intersect the interior). The line segment between two visible points is called a \emph{visibility edge}. The \emph{visibility graph} of a given point set $V$ and a given set of polygonal obstacles $S$, denoted $Vis(V,S)$, is the complete graph on $V$ excluding all the edges that properly intersect some obstacle (see Figure~\ref{fig:vis}). It is a well-known fact that the visibility graph is connected. 

	\begin{figure}[h!]
		\centering
		\includegraphics{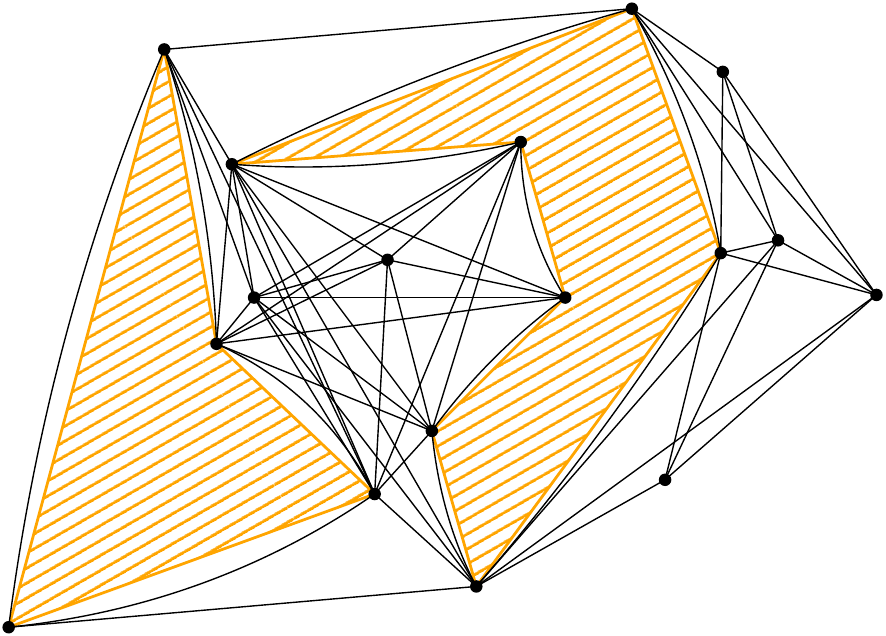}
		\caption{An example of a visibility graph. Edges coinciding with obstacle boundaries are drawn as slight arcs for clarity.}
		\label{fig:vis}
	\end{figure}

Clarkson~\cite{Cl87} was one of the first to study this problem, showing how to construct a $(1+\epsilon)$-spanner of $Vis(V,S)$. Modifying this result, Das~\cite{D97} showed that it is also possible to construct a spanner of constant spanning ratio and constant degree. Recently, Bose~\etal~\cite{basepaper} constructed a 6-spanner of degree $6+c$, where $c$ is the number of line segment obstacles incident to a vertex. In the process, they also show how to construct a 2-spanner of the visibility graph. We generalize these results and construct a 6-spanner of degree at most 7 in the presence of polygonal obstacles, simplifying some of the proofs in the process. Leading up to this main result, we first construct the \longname, denoted $G_\infty$ (defined in the next section), a 2-spanner of the visibility graph of unbounded degree. We modify this graph in a sequence of three steps, each giving a plane 6-spanner of the visibility graph to bound the degree to 15, 10, and finally 7. 

Each of these graphs may be of independent interest. Specifically, the graphs with degree 10 and 15 are constructed by solely removing edges from $G_\infty$. Furthermore, in the graph of degree 15 whether an edge of $G_\infty$ is kept can be determined by its endpoints, whereas for the other two this involves their neighbors. Hence, depending on the network model and communication cost, one graph may be more easily applicable than another.

\section{Preliminaries}
\label{sec:prelim}
Throughout this paper, we assume that each vertex is part of at most one polygonal obstacle and occurs at most once along its boundary, i.e., we have vertex-disjoint simple polygons. We note that we can always duplicate the vertex and possibly split the polygon, depending on whether we allow a path to go through this vertex in order to reach the opposite side of the polygon. An example of this is shown in Figure~\ref{fig:fish}. This operation does not affect the spanning ratio of the resulting graph, and the effect on the degree bound is minor: using Lemma~5 from~\cite{basepaper}, it can be shown that the degree of a vertex $v$ of the original graph is 6 plus half the number of polygon edges incident on $v$. 

\begin{figure}[h!]
	\centering
	\includegraphics{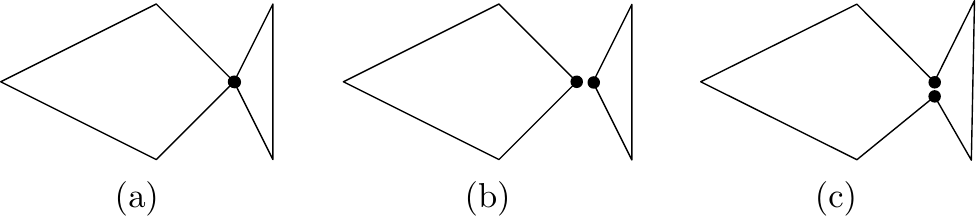}
	\caption{(a) Two polygons share a vertex. (b) Modified input when paths are allowed to pass between the polygons. (c) Modified input when paths are not allowed to pass between the polygons.}
	\label{fig:fish}
\end{figure}

Next, we describe how to construct the \longname, $G_\infty$. Before we can construct this graph, we first partition the plane around for each vertex $u \in V$ into six cones, each with angle $\theta=\frac{\pi}{3}$ and $u$ as the apex. For ease of exposition, we assume that the bisector of one cone is a vertical ray going up from $u$. We refer to this cone as $C^u_0$ or $C_0$ when the apex $u$ is clear from the context. The cones are then numbered in counter-clockwise order $(C_0, \overline{C_2}, C_1, \overline{C_0}, C_2, \overline{C_1})$ (see Figure~\ref{fig:cones}). Cones of the form $C_i$ are called \emph{positive cones}, whereas  cones of the form $\overline{C^u_i}$ are \emph{negative cones}. For any two vertices $u$ and $v$ in $V$, we have the property that if $v\in C_i^u$ then $u\in\overline{C^v_i}$. 

\begin{figure}[ht]
  \begin{minipage}[t]{0.47\linewidth}
    \begin{center}
      \includegraphics{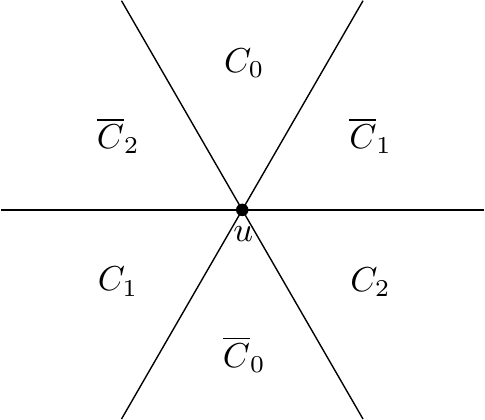}
    \end{center}
    \caption{The cones around $u$.}
    \label{fig:cones}
  \end{minipage}
  \hspace{0.05\linewidth}
  \begin{minipage}[t]{0.47\linewidth}
    \begin{center}
      \includegraphics{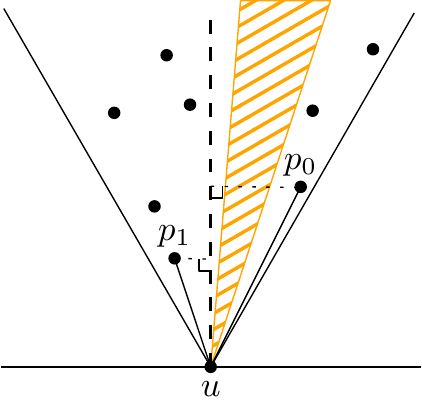}
    \end{center}
    \caption{The positive cone $C^u_i$ is split into two subcones and an edge is added in each.}
    \label{fig:subcones}
  \end{minipage}
\end{figure}

We are now ready to construct $G_\infty$. For every vertex $u \in V$, we consider each of its positive cones and add a single edge in such a cone. More specifically, we consider all vertices visible to $u$ that lie in this cone and add an undirected edge to the vertex whose projection on the bisector of the cone is closest to $u$. More precisely, the edge $(u,v)$ is part of $G_\infty$ when $v$ is visible to $u$ and $|u v'| < |u w'|$ for all $v \neq w$, where $v'$ and $w'$ are the projections of $v$ and $w$ on the bisector of the currently considered positive cone of $u$. For ease of exposition, we assume that no two vertices lie on a line parallel to one of the cone boundaries and no three vertices are collinear. This ensures that each vertex lies in a unique cone of each other vertex and their projected distances are distinct. If a point set is not in general position, one can rotate it by a small angle such that the resulting point set is in general position. 

Since every vertex is part of at most one obstacle, obstacles can affect the construction in only a limited number of ways. Cones that are split in two (i.e., there are visible vertices on both sides of the obstacle) are considered to be two subcones of the original cone and we add an edge in each of the two subcones using the original bisector (see Figure~\ref{fig:subcones}). If a cone is not split, the obstacle only changes the region of the cone that is visible from $u$. Since we only consider visible vertices when adding edges, this is already handled by the construction method. 

We note that the construction described above is similar to that of the constrained half-$\Theta_6$-graph as defined by Bose~\etal~\cite{basepaper} for line segment constraints. In their setting a cone can be split into multiple subcones and in each of the positive subcones an edge is added. This similarity will form a crucial part of the planarity proof in the next section. 

Before we prove that $G_\infty$ is a plane spanner, however, we first prove a useful visibility property that will form a building block for a number of the following proofs. We note that this property holds for any three points, not just vertices of the input. 

\begin{lemma}\label{lem:lem1}
	Let $u$, $v$, and $w$ be three points where $(w, u)$ and $(u, v)$ are both visibility edges and $u$ is not a vertex of any polygonal obstacle $P$ where the open polygon $P'$ of $P$ intersects $\triangle wuv$. The area $A$, bounded by $(w,u)$, $(u,v)$, and a convex chain formed by visibility edges between $w$ and $v$ inside $\triangle wuv$, does not contain any vertices and is not intersected by any obstacles. 
\end{lemma}
\begin{proof}
	We first construct the convex chain between $w$ and $v$. Let $Q$ be the set of vertices inside $\triangle wuv$. If $Q=\emptyset$, $(w,v)$ is the convex chain. If $Q\neq\emptyset$, consider the convex hull of $Q \cup \{v,w\}$. After removing edge $(w,v)$ we obtain a convex chain between $w$ and $v$ inside $\triangle wuv$. By construction, the area $A$ does not contain any vertices.
	
	\begin{figure}[h!]
		\centering
		\includegraphics{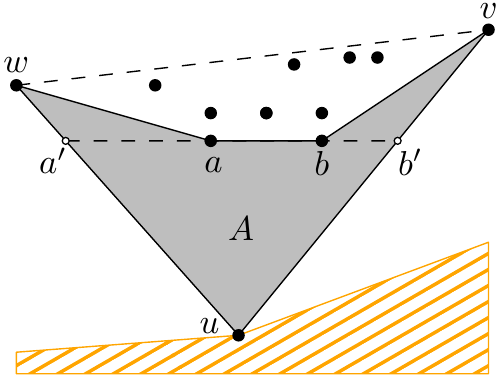}
		\caption{Triangle $wuv$ and the convex chain between $w$ and $v$ with consecutive vertices $a$ and $b$ and their intersections with the triangle's boundary at $a'$ and $b'$. The gray region $A$ is empty.}
		\label{fig:lem1}
	\end{figure}
	
	Next, we show, by contradiction, that every edge between two consecutive vertices $a$ and $b$ on the convex chain is a visibility edge. Extend $(a,b)$ to the boundaries of the triangle and let the intersections be $a'$ and $b'$ (see Figure~\ref{fig:lem1}). Since $u$ is not a vertex of any (open) polygonal obstacle intersecting $\triangle wuv$ and $(w, u)$ and $(u, v)$ are visibility edges, any polygonal obstacle crossing $(a,b)$ will have at least one vertex inside $\triangle a'ub'$. However, since $\triangle a'ub'$ is contained in $A$, this contradicts that $A$ is empty. Therefore, the convex chain consists of visibility edges. 
	
	Finally, since $A$ does not contain any vertices and is bounded by visibility edges, the vertices of any polygonal obstacle intersecting this area have to be contained in the set of vertices bounding $A$. To also intersect $A$, this implies that $u$ is one of these vertices, which is contradicts that $u$ is not a vertex of any such polygon. Hence, $A$ does not contain any polygonal obstacles. 
\end{proof}

For ease of notation, we define the canonical triangle of two vertices $u$ and $v$ with $v \in C_i^u$, denoted $\bigtriangledown_u^v$, to be the equilateral triangle defined by the boundaries of $C_i^u$ and the line through $v$ perpendicular to the bisector of $C_i^u$.

\section{The Polygon-Constrained-Half-$\boldsymbol{\Theta_6}$-Graph}
In this section we show that graph $G_\infty$ is a plane $2$-spanner of the visibility graph. We first prove it is plane.

\begin{lemma}\label{lem:lem3}
	\graphname is a plane graph.
\end{lemma}
\begin{proof}
	We prove this lemma by proving that \graphname is a subgraph of the constrained half-$\Theta_6$-graph introduced by Bose~\etal~\cite{basepaper}. Recall that in their graph the set of obstacles is a plane set of line segments. 
	
	\begin{figure}[h!]
		\centering
		\subfloat[$G_\infty$ with one obstacle $P$.]{
			\includegraphics{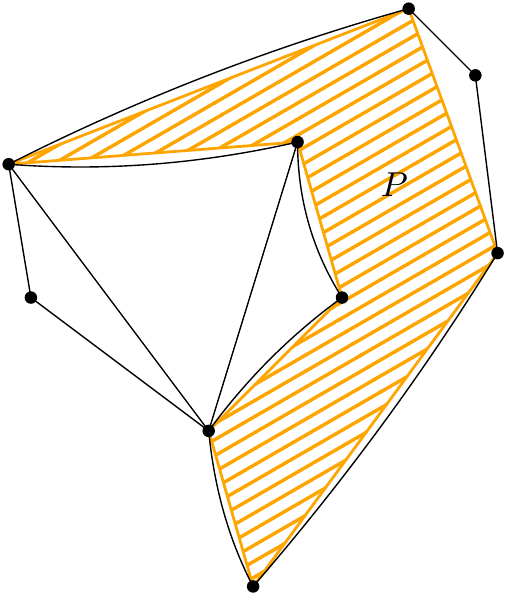}
		} \hspace{3em}
		\subfloat[The constrained half-$\Theta_6$-graph with constraints $l_0, l_1, ..., l_6$.]{
			\includegraphics{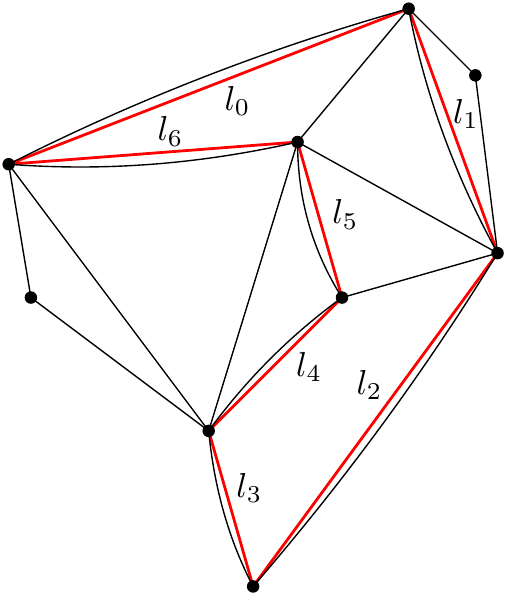}
		}
		\caption{\graphname is a subgraph of the constrained half-$\Theta_6$-graph. Edges coinciding with obstacle boundaries are drawn as slight arcs for clarity.}
		\label{fig:lem3}
	\end{figure}

	Given a set of polygonal obstacles, we convert them into line segments as follows: Each boundary edge on polygonal obstacle $P \in S$ forms a line segment obstacle $l_i$. Since these edges meet only at vertices and no two obstacles intersect, this gives a plane set of line segments. Recall that the constrained half-$\Theta_6$-graph constructs subcones the same way as \graphname does. This means that when considering the plane minus the interior of the obstacles in $S$, the constrained half-$\Theta_6$-graphs adds the same edges as $G_\infty$, while inside $P$ the constrained half-$\Theta_6$-graph may add additional edges (see Figure~\ref{fig:lem3}). Hence, $G_\infty$ is a subgraph of the constrained half-$\Theta_6$-graph. Since Bose~\etal showed that their graph is plane, it follows that $G_\infty$ is plane as well. 
\end{proof}

Next, we show that \graphname is a $2$-spanner of the visibility graph.

\begin{lemma}\label{lem:Ginftyspanning}
	Let $u$ and $v$ be vertices where $(u,v)$ is a visibility edge and $v$ lies in a positive cone of $u$. Let $a$ and $b$ be the two corners of $\bigtriangledown_u^v$ opposite to $u$ and let $m$ be the midpoint of $(a,b)$. There exists a path from $u$ to $v$ in \graphname such that the path is at most $(\sqrt{3}\cdot\cos{\angle muv} + \sin{\angle muv)}\cdot|uv| \leq 2\cdot|uv|$ in length.
\end{lemma}
\begin{proof}
	The following proof assumes without loss of generality that $v$ lies in $C^u_0$ (Figure~\ref{fig:lem4}a). We prove the lemma by induction on the area of $\bigtriangledown_u^v$ (or the ordering of all triangles $\bigtriangledown_u^v$ to be more precise). Consider $\bigtriangledown_u^v$ for every visibility edge $(u,v)$ and let $\delta(u,v)$ denote the length of the shortest path from $u$ to $v$ in \graphname that lies within $\bigtriangledown_u^v$. Let $A$ be the triangle $\triangle uav$ and $B$ be $\triangle ubv$ (see Figure~\ref{fig:lem4}b). In the following, we say that $A$ and $B$ are empty if they do not contain vertices.  We use the following induction hypothesis. 
	
	\begin{figure}[h!]
		\centering
		\subfloat[$\bigtriangledown_u^v$ with $a$, $b$, and $m$.]{
			\includegraphics{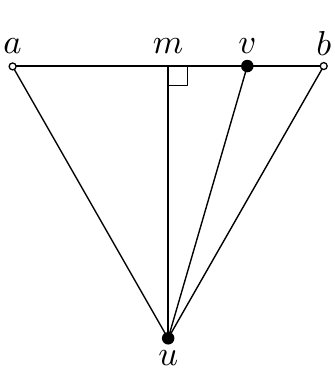}
		} \hspace{3em}
		\subfloat[Triangles $A$ and $B$.]{
			\includegraphics{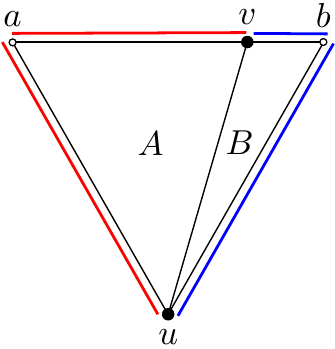}
		} \hspace{3em}
		\subfloat[The first step of the inductive argument.]{
			\includegraphics{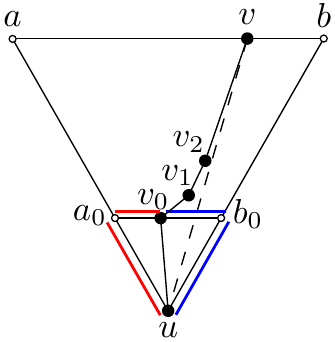}
		}
		\caption{The various points and region defined with respect to $\bigtriangledown_u^v$.}
		\label{fig:lem4}
	\end{figure}
	
	\textbf{Induction hypothesis:} 
	\[
	\delta(u,v)\leq 
	\begin{cases}
	\color{blue}|ub|+|bv|\color{black},& \text{if }A=\emptyset\\
	\color{red}|ua|+|av|\color{black},& \text{if }B=\emptyset\\
	\max\{\color{red}|ua|+|av|\color{black}, \color{blue}|ub|+|bv|\color{black}\},& \text{if }A\neq\emptyset \text{ and }B\neq\emptyset
	\end{cases}
	\] 
	
	\textbf{Base case:} If $\bigtriangledown_u^v$ is the smallest canonical triangle, $v$ is the closest visible vertex to $u$ in that positive subcone and thus $(u,v)$ is an edge in \graphname and $\delta(u,v)=|uv|$. Therefore, by triangle inequality, the inductive hypothesis holds. 
	
	\textbf{Induction step:} Assume that the induction hypothesis holds for every visibility edge with canonical triangle smaller than $\bigtriangledown_u^v$. If $(u,v)$ is an edge in the \longname, the induction hypothesis holds with the same argument as in the base case. Otherwise, let $v_0$ be the closest visible vertex of $u$ in the subcone of $C_0^u$ that contains $v$. This defines a canonical triangle $\bigtriangledown_u^{v_0}$ with $a_0$ and $b_0$ at its corners opposite to $u$ (see Figure~\ref{fig:lem4}c). Thus, $\delta(u,v)\leq|uv_0|+\delta(v_0,v)$ and  $|uv_0|\leq\min\{\color{red}|ua_0|+|a_0v_0|\color{black}, \color{blue}|ub_0|+|b_0v_0|\color{black}\}$. Without loss of generality, assume that $(u,v_0)$ lies to the left of $(u,v)$ and hence $A\neq\emptyset$. 
	
	Applying Lemma~\ref{lem:lem1} on triangle $\triangle uvv_0$ with visibility edges $(u,v)$ and $(u,v_0)$, there exists a convex chain, $v_0,...,v_k,v$ of visibility edges in $\triangle uvv_0$ between $u$ to $v$. Note that since $v_0$ lies in the same subcone as $v$, $u$ cannot be a vertex of an obstacle intersecting the triangle, as required for the application of Lemma~\ref{lem:lem1}. Since $v_0$ is the closest visible vertex to $u$, the vertices $v_1,...,v_k,v$ on the convex chain are all above the line $(a_0,b_0)$.
	
	Let $a_i$ and $b_i$ be the corners of canonical triangle of $v_{i-1}$ and $v_i$. Let $m_i$ be the midpoint of $(a_i,b_i)$ and define $A_i=\triangle a_iv_{i-1}v_i$ and $B_i=\triangle b_iv_{i-1}v_i$. There are three possible arrangements for each visibility edge $(v_{i-1},v_i)$ on the convex chain: (a) $v_{i-1}$ lies in the cone $C^{v_i}_1$, (b) $v_i$ lies in cone $C^{v_{i-1}}_0$ to the right of or on $m_i$, and (c) $v_i$ lies to the left of $m_i$ in cone $C^{v_{i-1}}_0$ (as shown in Figure~\ref{fig:lem4_1}). We proceed to bound the length of the induction path in each of these cases. 
	
	\begin{figure}[h!]
		\centering
		\subfloat[$v_{i-1}$ lies in the cone $C^{v_i}_1$.]{
			\includegraphics{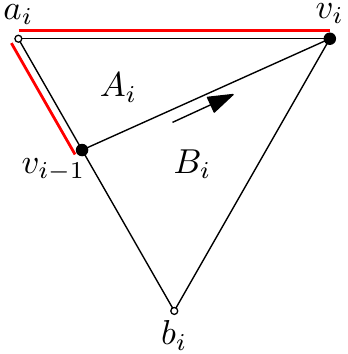}
		} \hspace{3em}
		\subfloat[$v_i$ lies in cone $C^{v_{i-1}}_0$ to the left of $m_i$.]{
			\includegraphics{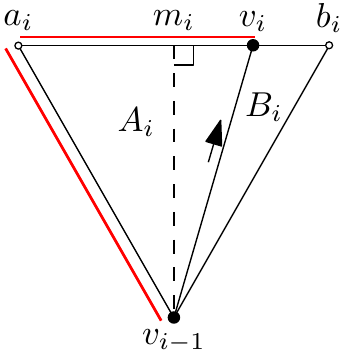}
		} \hspace{3em}
		\subfloat[$v_i$ lies in cone $C^{v_{i-1}}_0$ to the right of $m_i$.]{
			\includegraphics{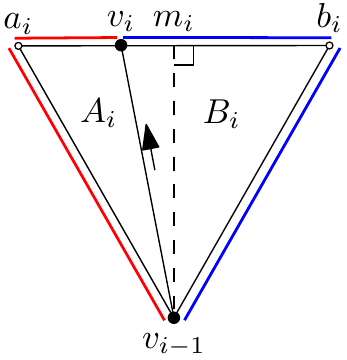}
		}
		\caption{The three configurations of visibility edges on the convex chain.}
		\label{fig:lem4_1}
	\end{figure}
	
	\textbf{Type (a):} Since $B_i$ is contained in the empty region of Lemma~\ref{lem:lem1} (it lies between the convex chain and $(u,v)$), it is empty. Since $(v_{i-1},v_i)$ is a visibility edge and the area of $\bigtriangledown_{v_i}^{v_{i-1}}$ is smaller than the area of $\bigtriangledown_u^v$, the induction hypothesis implies that $\delta(v_{i-1},v_i)\leq\color{red}|v_{i-1}a_i|+|a_iv_i|\color{black}$.
	
	\textbf{Type (b):} Since $(v_{i-1},v_i)$ is a visibility edge and the area of $\bigtriangledown_{v_{i-1}}^{v_i}$ is smaller than the area of $\bigtriangledown_u^v$, the induction hypothesis implies that $\delta(v_{i-1},v_i)\leq\max\{\color{red}|v_{i-1}a_i|+|a_iv_i|\color{black}, \color{blue}|v_{i-1}b_i|+|b_iv_i|\color{black}\}$. Since $v_i$ lies to the right of or on $m_i$, $\color{red}|a_iv_i|\color{black}\geq\color{blue}|b_iv_i|\color{black}$ and therefore $\color{red}|v_{i-1}a_i|+|a_iv_i|\color{black}\geq\color{blue}|v_{i-1}b_i|+|b_iv_i|\color{black}$ and $\delta(v_{i-1},v_i)\leq\color{red}|v_{i-1}a_i|+|a_iv_i|\color{black}$.
	
	\textbf{Type (c):} Since $(v_{i-1},v_i)$ is a visibility edge and the area of $\bigtriangledown_{v_{i-1}}^{v_i}$ is smaller than the area of $\bigtriangledown_u^v$, we can apply the induction hypothesis. If $B_i=\emptyset$ then $\delta(v_{i-1},v_i)\leq\color{red}|v_{i-1}a_i|+|a_iv_i|\color{black}$. Otherwise, $\delta(v_{i-1},v_i)\leq\color{blue}|v_{i-1}b_i|+|b_iv_i|\color{black}$, since $\color{red}|a_iv_i|\color{black}<\color{blue}|b_iv_i|\color{black}$.

	Note that the three types of visibility edges appear in order on the convex chain when directed from $u$ to $v$. To complete the proof, we consider the three possible convex chains, depending on the location of $v$ with respect to $m$ and whether $B$ is empty. 
	
	\begin{figure}[h!]
		\centering
		\subfloat[$v$ lies on segment $m b$.]{
			\includegraphics{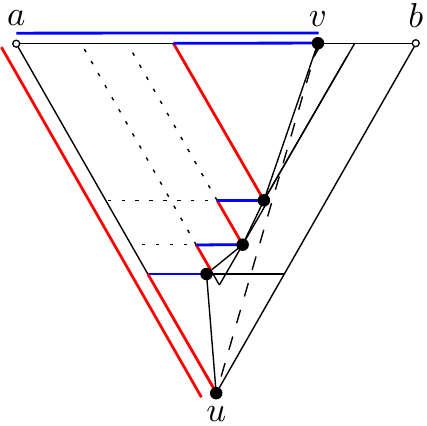}
		} \hspace{3em}
		\subfloat[$v$ lies on segment $a m$ and $B = \emptyset$.]{
			\includegraphics{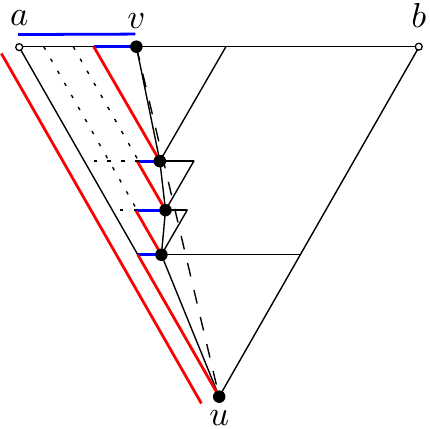}
		} \hspace{5em}
		\subfloat[$v$ lies on segment $a m$ and $B \neq \emptyset$. Bounding the two parts of the convex chain.]{
			\includegraphics{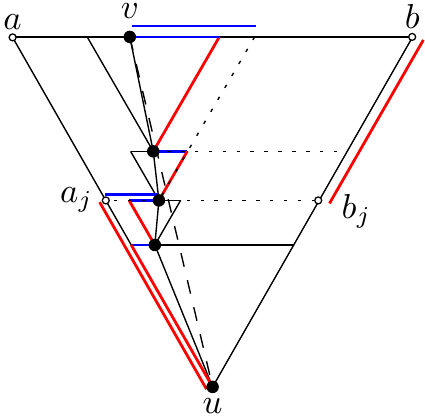}
		} \hspace{3em}
		\subfloat[$v$ lies on segment $a m$ and $B \neq \emptyset$. Summing up to obtain the total path length.]{
			\includegraphics{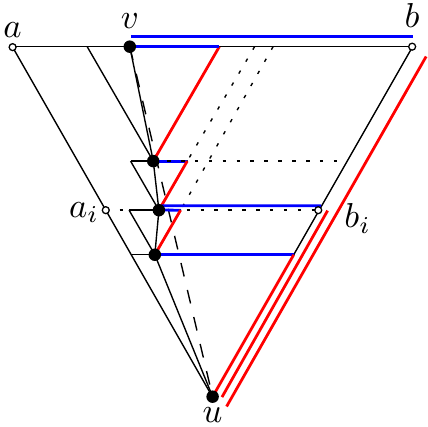}
		}
		\caption{Bounding the total length of the three possible compositions of the convex chain.}
		\label{fig:lem4_2}
	\end{figure}
	
	\textbf{Case (a):} If $v$ lies on $m b$, the convex chain can only contain Type (a) and Type (b) configurations. We can bound the total length of the induction path by summing up the two components of the inductive lengths of each visibility edge $(v_{i-1},v_i)$ (see Figure~\ref{fig:lem4_2}a). The horizontal components sum up to $|av|$ and the components parallel to $ua$ sum up to $|ua|$. Hence, we can conclude that $\delta(u,v)\leq\color{red}|ua|\color{black}+\color{blue}|av|\color{black}$.
	
	\textbf{Case (b):} If $v$ lies on $a m$ and $B=\emptyset$, the convex chain can contain all three types of visibility edges. However, since $B$ is empty and using the empty area implied by Lemma~\ref{lem:lem1}, the area between the convex chain and $(u,v)$ is empty. Hence, $B_i$ is empty for each canonical triangle along the convex chain. Therefore, for visibility edges of Type (c), $\delta(v_{i-1},v_i)\leq|v_{i-1}a_i|+|a_iv_i|$. The total length of the path can now be bounded the same way as was done in the previous case (see Figure~\ref{fig:lem4_2}b). Thus, $\delta(u,v)\leq\color{red}|ua|\color{black}+\color{blue}|av|\color{black}$. 
	
	\textbf{Case (c):} If $v$ lies on $a m$ and $B\neq\emptyset$, the convex chain can again contain all three types of visibility edges. Since $v_0$ lies in $A$, both $A$ and $B$ are not empty. Thus, since $|a v| < |v b|$, it suffices to show that $\delta(u,v)\leq\color{red}|ub|\color{black}+\color{blue}|bv|\color{black}$. Let $(v_j, v_{j+1})$ be the first Type (c) visibility edge and let $a_j$ and $b_j$ be the two corners of $\bigtriangledown_u^{v_j}$. We can sum up the length of the path up to $v_j$ the same way we did in case (a), since this part only contains Type (a) and Type (b) visibility edges (see Figure~\ref{fig:lem4_2}c). This gives that this part of the path has length at most $|u a_j| + |a_j v_j|$. Next, since $v$ lies on $a m$, $|a_jv_j|<|b_jv_j|$. Thus, the length of the first part of the path is at most $|u b_j| + |b_j v_j|$. Adding to this the lengths of all Type~(c) confirgurations, we obtain the following bound: $\delta(u,v)\leq\color{red}|ub|\color{black}+\color{blue}|bv|\color{black}$ (see Figure~\ref{fig:lem4_2}d).

	Thus, in all cases the induction hypothesis is satisfied. Using basic trigonometry, we can express the maximum value of the induction hypothesis as $\sqrt{3}\cdot\cos \angle muv+\sin \angle muv$, for $\angle muv\in[0,\pi/6]$. This function is maximal when $\angle muv = \pi/6$, where it has value $2$. Therefore, \longname is a spanner with spanning ratio $2$.
\end{proof}

Lemmas~\ref{lem:lem3} and~\ref{lem:Ginftyspanning} imply the following theorem. 

\begin{theorem}
$G_\infty$ is a plane $2$-spanner of the visibility graph. 
\end{theorem}

We note that while every vertex in $G_\infty$ has at most one edge in every positive subcone, it can have an unbounded number of edges in its negative subcones. In the following sections, we proceed to bound the degree of the spanner.

\section{Bounding the Degree}
In this section, we introduce $G_{15}$, a subgraph of the \longname of maximum degree 15. We obtain $G_{15}$ from $G_\infty$ by, for each vertex, removing all edges from its negative subcones except for the leftmost (clockwise extreme from $u$'s perspective) edge, rightmost (counterclockwise) edge, and the edge to the closest vertex in that cone (see Figure~\ref{fig:G15}a). Note that the edge to the closest vertex may also be the leftmost and/or the rightmost edge. 
	
	\begin{figure}[h!]
		\centering
		\subfloat[In $\overline{C^u_0}$ only the edges to $v_0$, $v_3$, and $v_4$ are kept.]{
			\includegraphics{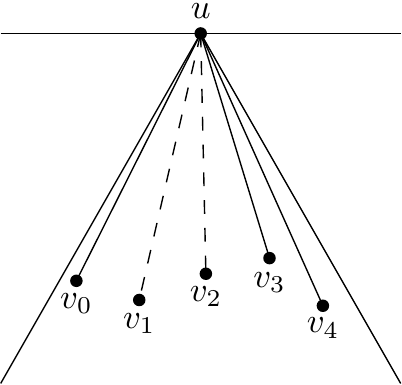}
		} \hspace{3em}
		\subfloat[The canonical path in $\overline{C^u_0}$ with canonical sequence $(v_0,v_1, v_2,v_3,v_4)$.]{
			\includegraphics{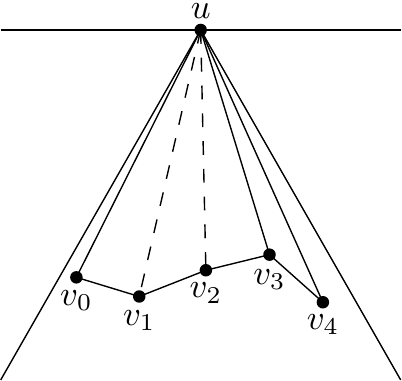}
		}
		\caption{Transforming $G_\infty$ into $G_{15}$.}
		\label{fig:G15}
	\end{figure}

By simply counting the number of subcones (three positive, three negative, and at most one additional subcone caused by an obstacle), we obtain the desired degree bound. Furthermore, since $G_{15}$ is a subgraph of $G_\infty$, it is also plane. 

\begin{lemma}
	The degree of each vertex $u$ in $G_{15}$ is at most $15$.
\end{lemma}
\begin{proof}
	By construction, there could only be one edge lying in each positive subcone of vertex $u$. During the transformation to $G_{15}$, we removed all edges except at most three edges in each negative subcone of $u$. Temporarily ignoring obstacles, there are three positive and three negative cones and therefore, the degree bound without obstacles is 12. 
	
	\begin{figure}[h!]
		\centering
		\includegraphics{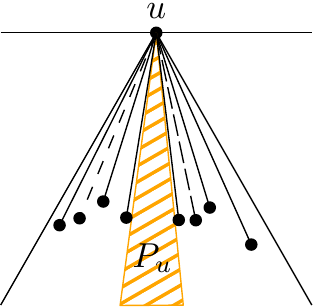}
		\caption{Negative cone split by obstacle $P_u$ has two subcones and at most 6 edges.}
		\label{fig:g15bound}
	\end{figure}
	
	Recall that each vertex $u\in V$ is part of at most one obstacle and denote this obstacle by $P_u$. If $P_u$ lies across multiple cones, the corresponding cones will only shrink in size and the number of subcones of $u$ will not increase. Otherwise, $P_u$ splits a cone into two subcones. If $P_u$ lies within $C^u_i$, this creates an extra positive subcone with one edge and the degree bound will increases to 13. If $P_u$ lies within $\overline{C^u_i}$, this creates an extra negative subcone with at most three edges and the degree bound increases to 15 (see Figure~\ref{fig:g15bound}), completing the proof. 
\end{proof}

\begin{lemma}\label{lem:G15P}
	$G_{15}$ is plane.
\end{lemma}
\begin{proof}
	$G_{15}$ is constructed by removing edges from the negative subcones of vertices in $G_\infty$. Therefore, $G_{15}$ is a subgraph of $G_\infty$. Since by Lemma~\ref{lem:lem3} $G_\infty$ is a plane graph, $G_{15}$ is also a plane graph.
\end{proof}

Finally, we look at the spanning ratio of $G_{15}$. We assume without loss of generality that we look at two vertices $u$ and $v$ such that $v\in\overline{C^u_0}$ and $u\in C^v_0$. Given vertex $u$ and a negative subcone, we define the canonical sequence of this subcone as the vertices adjacent to $u$ in \graphname that lie in the subcone in counterclockwise order (see Figure~\ref{fig:G15}b). The canonical path of $\overline{C^u_0}$ refers to the path that connects the consecutive vertices in the canonical sequence of $\overline{C^u_0}$. This definition is not dissimilar to that of Bose~\etal~\cite{basepaper}. 

\begin{lemma}\label{lem:canonpathempty}
	If $v_i$ and $v_{i+1}$ are consecutive vertices in the canonical sequence of the negative subcone $\overline{C^{u}_0}$, then $\triangle uv_iv_{i+1}$ is empty.
\end{lemma}
\begin{proof}
	By Lemma~\ref{lem:lem3}, \graphname is a plane graph and by construction ${(u,v_i),(u,v_{i+1})}\in$ \graphname, therefore no edges or polygonal obstacles can intersect these sides of $\triangle uv_iv_{i+1}$. Hence, any obstacle or edge intersecting $\triangle uv_iv_{i+1}$ has at least one vertex within the triangle. 
	
	Let $A$ be the intersection of $\triangle uv_iv_{i+1}$ with the positive subcones of $v_i$ and $v_{i+1}$ that contain $u$ (see Figure~\ref{fig:empty_triangle}a). We prove by contradiction that $A=\emptyset$. Let $x$ be the highest vertex in $C^{v_i}_0\cap A$. Since $(u,v_i)$ and $(u,v_{i+1})$ exist in \graphname, $(u,v_i)$ is a visibility edge and any obstacles blocking the visibility between $u$ and $x$ has to have at least a vertex above $x$ in $C^{v_i}_0\cap A$. Therefore, $(u,x)$ is a visibility edge. By applying Lemma~\ref{lem:lem1} to $\triangle uv_ix$, we obtain a convex chain from $v_i$ to $x$. Since the neighbor of $v_i$ along the convex chain lies in $C^{v_i}_0$ and is visible to $v_i$, $u$ cannot be the closest visibility vertex of $v_i$, contradicting that $(u,v_i)\in$ \graphname. An analogous argument shows that $C^{v_{i+1}}_0\cap A$ is empty and hence, $A$ is empty.
	
	\begin{figure}[h!]
		\centering
		\subfloat[Area $A$ is empty.]{
			\includegraphics{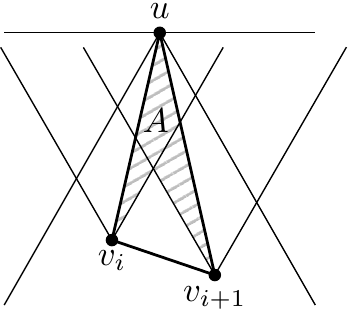}
		} \hspace{3em}
		\subfloat[Area $B$, $\triangle uv_iv_{i+1} \setminus A$, is empty.]{
			\includegraphics{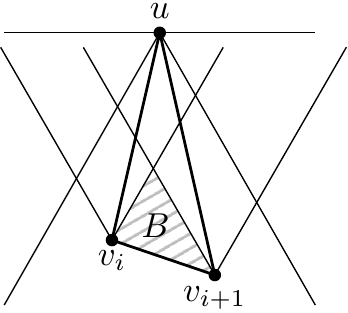}
		}
		\caption{Area $A$ and $B$ are both empty, hence $\triangle uv_iv_{i+1}$ is empty.}
		\label{fig:empty_triangle}
	\end{figure}
	
	Triangle $\triangle uv_iv_{i+1} \setminus A$ is denoted as area $B$. We show by contradiction that $B$ is also empty. Let $y\in B$ be the highest vertex in $B$ (see Figure~\ref{fig:empty_triangle}b). Since $A=\emptyset$ and no edges or obstacles pass through $(u,v_i)$ and $(u,v_{i+1})$, $(u,y)$ is a visibility edge that lies within $\overline{C^{u}_0}$. Since $y\in B\subseteq\triangle uv_iv_{i+1}$, $C^y_0$ lies within $(C^{v_i}_0\cup C^{v_{i+1}}_0\cup B)$. As $u$ is the closest visible vertex in both $C^{v_i}_0$ and $C^{v_{i+1}}_0$, $u$ is also the closest visible vertex in $C^y_0$ and therefore, $(u,y)$ exists in \graphname. This implies that $y$ is in the canonical sequence between $v_i$ and $v_{i+1}$, contradicting that $v_i$ and $v_{i+1}$ are consecutive. Hence, $B=\emptyset$.
	
	Since $A=\emptyset$ and $B=\emptyset$, $\triangle uv_iv_{i+1}=A\cup B=\emptyset$. Finally, since any obstacle being fully contained in $\triangle uv_iv_{i+1}$ would imply that it contains some vertices as well, it is empty of both vertices and obstacles. 
\end{proof}

\begin{lemma}\label{lem:canonpath}
	$G_{15}$ contains an edge between every pair of consecutive vertices on a canonical path.
\end{lemma}
\begin{proof}
	We first prove that the edges on the canonical path are in \graphname. Let $v_1$ and $v_2$ be a pair of consecutive vertices in the canonical sequence in $\overline{C^u_0}$. Assume, without loss of generality, that $v_2\in\overline{C^{v_1}_2}$ and $v_1\in C^{v_2}_2$. Let $A$ be the area bounded by the positive subcones $C^{v_1}_0$ and $C^{v_2}_0$ that contain $u$ (see Figure~\ref{fig:G15Span}a). Let $A'$ be the set of vertices visible to $v_1$ or $v_2$ in $A$. We first show by contradiction that $A'=\emptyset$. According to Lemma~\ref{lem:canonpathempty}, $\triangle uv_1v_2$ is empty and thus does not contain any vertices. We focus on the part of $A$ to the left of $u v_1$. Consider vertex $x\in A$ to the left of $u v_1$ where $x$ has the smallest interior angle $\angle xuv_1$. Since $(u,v_1)$ is a visibility edge, the edge $(u,x)$ is a visibility edge, as any obstacle blocking it implies the existence of a vertex with smaller angle. Applying Lemma~\ref{lem:lem1} to $x u v_1$, there exists a convex chain from $v_1$ to $x$ in $\triangle v_1xu$. The neighbor of $v_1$ in the convex chain is a closer visible vertex than $u$ in $C^{v_1}_0$, contradicting that $(u,v_1)\in$ \graphname. An analogous argument for the region to the right of $u v_2$ contradicts that $(u,v_2)\in$ \graphname. Therefore, $A$ does not contain any vertices visible to $v_1$ or $v_2$.
	
	\begin{figure}[h!]
		\centering
		\subfloat[Area $A$ is bounded by $C^{v_1}_0$ and $C^{v_2}_0$.]{
			\includegraphics{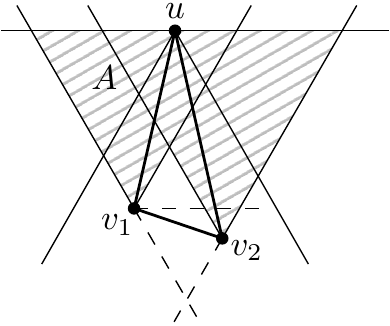}
		} \hspace{3em}
		\subfloat[Triangle $B$ is bounded by $\overline{C^{v_1}_1}$ and $C^{v_2}_0$.]{
			\includegraphics{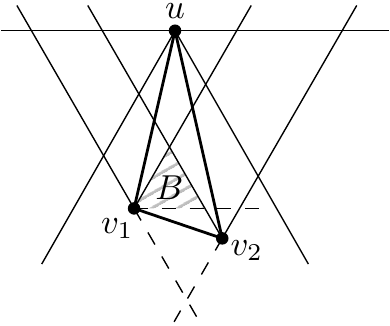}
		} \hspace{3em}
		\subfloat[Triangle $C$ is $\bigtriangledown_{v_1}^{v_2}$.]{
			\includegraphics{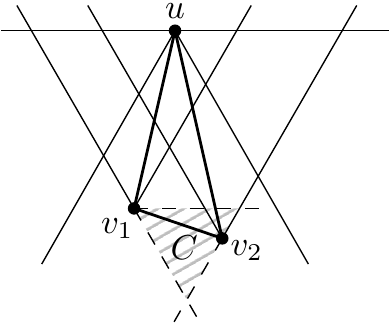}
		}
		\caption{A pair of consecutive vertices $v_1$ and $v_2$ along a canonical path and their surrounding empty areas.}
		\label{fig:G15Span}
	\end{figure}
	
	Next, let $B$ be the intersection of $\overline{C^{v_1}_1}$ and $\overline{C^{v_2}_2}$ (see Figure~\ref{fig:G15Span}b). Since $(v_1,v_2)$ is part of the canonical path of $\overline{C^u_0}$, by Lemma~\ref{lem:canonpathempty}, $\triangle uv_1v_2=\emptyset$. Since $B$ is contained in $\triangle uv_1v_2$, $B$ is also empty.
	
	Finally, let area $C$ be the canonical triangle $\bigtriangledown_{v_1}^{v_2}$ (see Figure~\ref{fig:G15Span}c). Let $C'$ be the set of visible vertices to $v_1$ and $v_2$ in $C$. We show that $C'=\emptyset$, by contradiction. Since both $A$ and $B$ are empty, the highest vertex $x\in C$ will have $u$ as the closest visible vertex in $C^x_0$. This implies that $x$ occurs between $v_1$ and $v_2$ in the canonical path, contradicting that they are consecutive vertices. Therefore, $C'=\emptyset$. 

	Since $A$, $B$, and $C$ are all empty, $v_2$ is the closest vertex of $v_1$ in the subcone of $C^{v_1}_2$ that contains it. By Lemma~\ref{lem:canonpathempty} it is visible to $v_1$ and thus, $(v_1, v_2)$ is an edge in \graphname. 
	
	It remains to show that $(v_1,v_2)$ is preserved in $G_{15}$. Since $\triangle uv_1v_2$ is empty and $G_{15}$ is plane (by Lemma~\ref{lem:G15P}), $v_1$ is the rightmost vertex in $\overline{C^{v_2}_2}$. Hence, it is not removed and thus, the canonical path in $\overline{C^u_0}$ exists in $G_{15}$.
\end{proof}

Now that we know that these canonical paths exist in $G_{15}$, we can proceed to prove that it is a spanner. 

\begin{lemma}\label{lem:G15spanning}
	$G_{15}$ is a $3$-spanner of \graphname.
\end{lemma}
\begin{proof} 
	Consider an edge $(u,v)$ in \graphname and assume, without loss of generality, that $v\in\overline{C^u_0}$ and let $v_0$ be the closest vertex that $u$ has an edge to in $\overline{C^u_0}$. Consider the part of the canonical path between $v_0$ and $v$ and denote its vertices by $v_0, v_1, ..., v_k=v$ (see Figure~\ref{fig:G15SR}). The path from $u$ via the vertices of the canonical path to $v$ is an upper bound of $\delta(u,v)$, the length of the shortest path from $u$ to $v$ in $G_{15}$. 
	
	\begin{figure}[h!]
		\centering
		\includegraphics{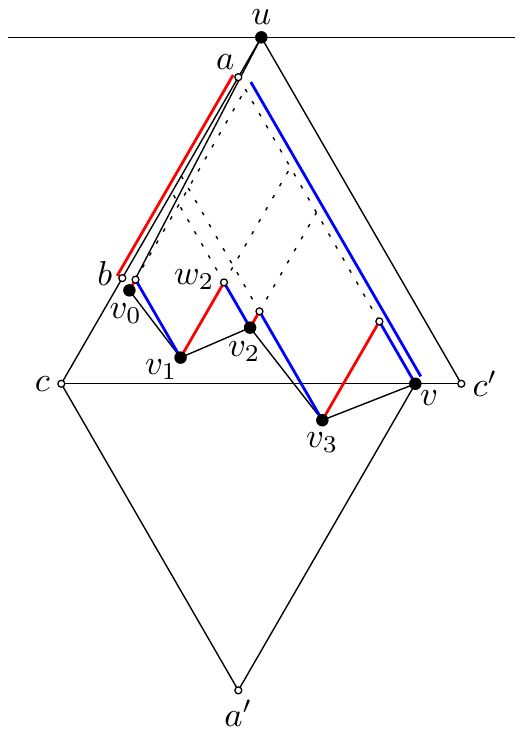}
		\caption{Negative cone $\overline{C^u_0}$ with canonical sequence $v_0, v_1, ..., v_k=v$.}
		\label{fig:G15SR}
	\end{figure}
	
	Let $c$ and $c'$ be the intersections of the cone boundaries of $\overline{C^u_0}$ with the line perpendicular to the bisector of $\overline{C^u_0}$ that intersects $v$. Let $a$ be the point on $uc$ such that $a v$ is parallel to $u c'$ and let $b$ be the point on $uc$ such that $b v_0$ is parallel to $u c'$. Consider the reflection $a'$ of $a$ over the line $c v$. We now show that no vertices on the canonical path between $v_0$ and $v$ lie outside $\triangle a'cv \cup \triangle acv$, since any vertex $x$, below $\triangle a'cv$ would have $v$ or $v_0\in C^x_0$, or occur before $v_0$ or after $v$ on the canonical sequence. If $x$ occurs before $v_0$ or after $v$ on the canonical sequence, it is not part of the canonical path between $v_0$ and $v$. If $v\in C^x_0$ is visible to $x$, it is closer than $u$ and $(u,x)$ could not exist in \graphname. If $v$ is not visible to $x$, then there exists a vertex $y$ of the obstacle $P$ blocking $v$ in $C^x_0$ (as $(u,x)$ needs to be visible for $x$ to lie on the canonical path of $u$, $P$ cannot block the cone entirely). Hence, $C^x_0$ contains vertices that are closer than $u$ and since the cone cannot be completely blocked, the vertex $z$ that minimizes the angle $\angle u x z$ is visible to $x$. Consequently, $(u,x)$ could not exist in \graphname. An analogous argument shows that if $v_0\in C^x_0$, $(u,x)$ does not exist. 

	It remains to bound the length of the canonical path. By triangle inequality, $|u v_0|\leq|ub|+|b v_0|$. For every consecutive pair of vertices $v_{i-1}, v_i$, let $w_i$ be the intersection of the right cone boundary of $C^{v_{i-1}}_0$ and the left cone boundary of $C^{v_i}_0$ (see Figure~\ref{fig:G15SR}). By triangle inequality, we have that $|v_{i-1} v_i|\leq\color{red}|v_{i-1} w_i|\color{black}+\color{blue}|w_i v_i|\color{black}$. Since all line segments of the form $v_{i-1} w_i$ are parallel to $a c$ and all line segments of the form $w_i v_i$ are parallel to $a v$, summing them up gives $\color{red}|a b|\color{black}-|b v_0|+\color{blue}|a v|\color{black}$. Since $v_0$ lies in $\triangle ucc'$, $|a b|\leq|u c|$. By construction $\triangle ucc'$ is an equilateral triangle and since $av$ is parallel to $u c'$, $\triangle acv$ is also an equilateral triangle and $|a c|=|c v|=|a v|$. 
	
	Adding the upper bound on $|u v_0|$ to the length of the canonical path, we get $\delta(u,v)\leq|u b|+|b v_0|+|a b|+|a v|-|b v_0|$. Since $|a b|\leq|a c|$, $|u b|\leq|u c|$, and $|a c|=|c v|=|a v|$, we can rewrite this to $\delta(u,v)\leq|u c|+2\cdot|c v|$. Using basic trigonometric functions, we get that $|u c| = (\cos \angle vuc + \sin \angle vuc / \sqrt{3}) \cdot |u v|$ and $|c v| = 2 \cdot \sin \angle vuc / \sqrt{3} \cdot |u v|$. This implies that $\delta(u,v)\leq(\cos{\angle vuc}+\frac{5}{\sqrt{3}} \cdot \sin{\angle vuc}) \cdot |u v|$ where $0<\angle vuc\leq\frac{\pi}{3}$. As this is an increasing function, it is maximal at $\angle vuc = \frac{\pi}{3}$, where it attains value 3. Therefore, $G_{15}$ is a $3$-spanner of \graphname.
\end{proof}

Since by Lemma~\ref{lem:Ginftyspanning}, \graphname is a $2$-spanner of the visibility graph and by Lemma~\ref{lem:G15spanning}, $G_{15}$ is a $3$-spanner of \graphname, we obtain the main result of this section.

\begin{theorem}
	$G_{15}$ is a plane $6$-spanner of the visibility graph of degree at most 15. 
\end{theorem}

\section{Improving the Analysis}
Observing that we need to maintain only the canonical paths between vertices, along with the edge to the closest vertex in each negative subcone for the proof of Lemma~\ref{lem:G15spanning} to hold, we reduce the degree as follows. Consider each negative subcone $\overline{C^u_i}$ of every vertex $u$ and keep only the edge incident to the closest vertex (with respect to the projection onto the bisector of the cone) and the canonical path in $\overline{C^u_i}$ (see Figure~\ref{fig:G10}). We show that the resulting graph, denoted $G_{10}$, has degree at most 10. Recall that at most one cone per vertex is split into two subcones and thus this cone can have two disjoint canonical paths, as we argue in Lemma~\ref{lem:G10Degree}. 

\begin{figure}[h!]
	\centering
	\includegraphics{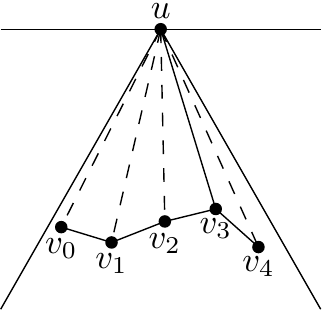}
	\caption{$G_{10}$ keeps only the canonical path and the edge to the closest adjacent vertex.}
	\label{fig:G10}
\end{figure}

As in the previous section, we proceed to prove that $G_{10}$ is a plane 6-spanner of the visibility graph of degree 10. Since $G_{10}$ is a subgraph of $G_{15}$, it is also plane. Furthermore, since the canonical path is maintained, the spanning property does not change. 

\begin{lemma}\label{lem:G10Plane}
	$G_{10}$ is a plane graph.
\end{lemma}
\begin{proof}
	Lemma~\ref{lem:canonpath} shows that $G_{15}$ contains an edge between consecutive vertices along each canonical path. Since $G_{10}$ consists of these canonical paths and the edge to the closest vertex, all of its edges are part of $G_{15}$. By Lemma~\ref{lem:G15P}, $G_{15}$ is a plane graph and thus $G_{10}$ is also plane.
\end{proof}

\begin{lemma}\label{lem:G10Spanning}
	$G_{10}$ is a $3$-spanner of \graphname.
\end{lemma}
\begin{proof}
	According to Lemma~\ref{lem:G15spanning}, the canonical paths in $G_{15}$ are the $3$-spanning paths for each edge in \graphname. Since $G_{10}$ preserves these canonical paths from $G_{15}$, these same paths in $G_{10}$ still serve as $3$-spanning paths for each edge in \graphname. Thus, $G_{10}$ is a 3-spanner of \graphname.
\end{proof}

It remains to upper bound the maximum degree of $G_{10}$. We do this by charging the edges incident to a vertex to its cones. The four part charging scheme is described below (see Figure~\ref{fig:G10Charges}). Scenarios A and B handle the edge to the closest vertex, while Scenarios C and D handle the edges along the canonical path. Hence, the total charge of a vertex is an upper bound on its degree. 
	
	\begin{figure}[h!]
		\centering
		\subfloat[Scenario A.]{
			\includegraphics{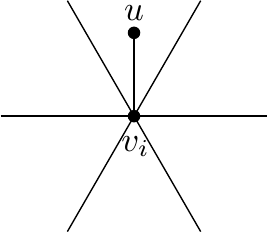}
		} \hspace{3em}
		\subfloat[Scenario B.]{
			\includegraphics{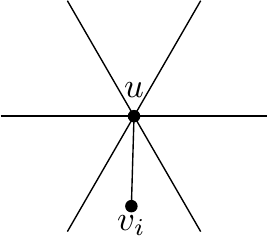}
		} \hspace{3em}
		\subfloat[Scenario C.]{
			\includegraphics{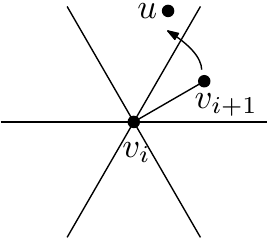}
		} \hspace{3em}
		\subfloat[Scenario D.]{
			\includegraphics{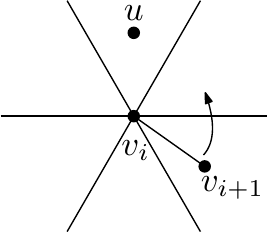}
		}
		\caption{The four scenarios of the charging scheme. The arrows in Scenarios C and D indicate the cone these edges are charged to.}
		\label{fig:G10Charges}
	\end{figure}

\textbf{Scenario A:} Edge $(u,v_i)$ lies in $C^{v_i}_j$ and $v_i$ is the closest vertex of $u$ in $\overline{C^{u}_j}$. Then $(u,v_i)$ is charged to $ C^{v_i}_j$. 

\textbf{Scenario B:} Edge $(u,v_i)$ lies in $\overline{C^u_j}$ and $v_i$ is the closest vertex of $u$ in $\overline{C^u_j}$. Then $(u,v_i)$ is charged in $\overline{C^u_j}$. 

\textbf{Scenario C:} Edge $(v_i,v_{i+1})$ lies on a canonical path of $u$, with $u\in C^{v_i}_j$, and $v_{i+1}\in\overline{C^{v_i}_{j+1}}$. Edge $(v_i,v_{i+1})$ is charged to $C^{v_i}_j$. Similarly, edge $(v_i,v_{i-1})$ on a canonical path of $u$, where $v_{i-1}\in\overline{C^{v_i}_{j-1}}$, is charged to $C^{v_i}_j$. 

\textbf{Scenario D:} Edge $(v_i,v_{i+1})$ lies on a canonical path of $u$, with $u\in C^{v_i}_j$, and $v_{i+1}\in C^{v_i}_{j-1}$. Edge $(v_i,v_{i+1})$ is charged to $\overline{C^{v_i}_{j+1}}$. Similarly, edge $(v_i,v_{i-1})$ on a canonical path of $u$, where $v_{i-1}\in C^{v_i}_{j+1}$, is charged to $\overline{C^{v_i}_{j-1}}$.

We note that every edge in $G_{10}$ is charged to both of its endpoints. Scenarios A and B consider an edge formed and preserved as the shortest edge, i.e., edge $(u,v_3)$ in Figure~\ref{fig:G10}. Scenarios C and D look at edges on canonical paths, where Scenario C handles edges in negative subcones and Scenario D takes care of the positive ones.

\begin{lemma}\label{lem:G10Degree}
	The degree of every vertex in $G_{10}$ is at most 10.
\end{lemma}
\begin{proof}
	\textbf{Negative subcones:} We first prove that there can be only one edge charged to each negative subcone. According to the charging scheme, only Scenario B and D charge edges to a negative subcone. We first consider Scenario B, where $u$ is the closest vertex of $v_i$ in its negative subcone $\overline{C^{v_i}_j}$. By definition of $G_{10}$, edge $(u,v_i)$ is the only edge preserved and charged $\overline{C^{v_i}_j}$ by this scenario. If $v_i$ is part of a polygonal obstacle $P$ and $P$ splits the negative cone (see Figure~\ref{fig:G10_Degree_B}), the number of negative subcones increases. Both negative subcones are charged 1 if they both have an edge belonging to scenario B. 
	
	\begin{figure}[h!]
		\centering
		\includegraphics{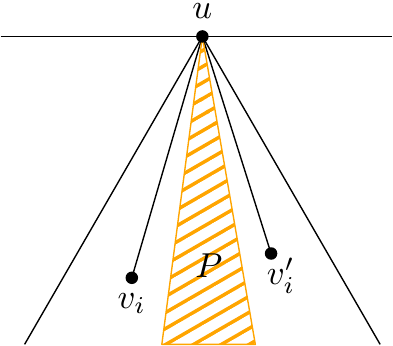}
		\caption{Two negative subcones can both be charged once from Scenario B.}
		\label{fig:G10_Degree_B}
	\end{figure}
	
	For Scenario D, the edge $(v_i,v_{i+1})$ lies in $C^{v_i}_{j-1}$ and is charged to $\overline{C^{v_i}_{j+1}}$. Since $(v_i,v_{i+1})$ is an edge in the canonical path of $\overline{C^u_{j}}$, there is only one edge from this canonical path in $C^{v_i}_{j-1}$ with $v_i$ as an endpoint. Next, we show that $\overline{C^{v_i}_{j+1}}$ is not charged by Scenario B, implying that its total charge is $1$. Consider triangle $\triangle uv_iv_{i+1}$.  Since this triangle was present in $G_\infty$, by Lemma~\ref{lem:canonpathempty} it contains no polygonal obstacles. We split $\overline{C^{v_i}_{j+1}}$ into two parts: $A$ lies outside $\triangle uv_iv_{i+1}$ and $B$ lies inside $\triangle uv_iv_{i+1}$ (see Figure~\ref{fig:G10_Degree_D}). 
	
	\begin{figure}[h!]
		\centering
		\includegraphics{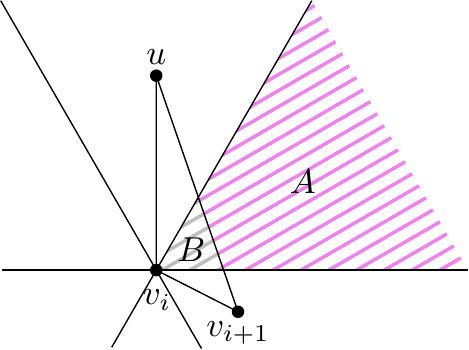}
		\caption{Area $A$ is shaded in gray and area $B$ is shaded in violet.}
		\label{fig:G10_Degree_D}
	\end{figure}
	
	Consider area $A$. Since edge $(v_{i+1},u)$ existed in \graphname, which by Lemma~\ref{lem:lem3} \graphname is plane, and $G_{10}$ is a subgraph of \graphname, there can be no edges between $v_i$ and any vertex in $A$. Thus, $v_i$ cannot be charged by Scenario B for any vertex in $A$. 
	
	Next, consider area $B$. Since $(v_i,v_{i+1})$ is on the canonical path of $\overline{C^{u}_j}$, both $(u,v_i)$ and $(u,v_{i+1})$ exist in \graphname by Lemma~\ref{lem:canonpath}. According to Lemma~\ref{lem:canonpathempty}, $\triangle uv_iv_{i+1}$ is empty and thus $B$ is empty as well. Thus, $v_i$ cannot be charged by Scenario B for any vertex in $B$ either. 
	
	Finally, since $B$ is empty and no polygonal obstacle can cross any side of $\triangle uv_iv_{i+1}$ (because of the planarity of \graphname), no edges can be charged by Scenario D if $v_i$ is a vertex of a polygonal obstacle $P$ with polygon boundary in $\overline{C^{v_i}_{j+1}}$. 

	Thus, if a negative subcone is charged with an edge by Scenario D, it cannot be charged with any edge by Scenario B. Conversely, if a negative subcone is charged with an edge by Scenario B, then it cannot be charged with any Scenario D edges either. Since both scenario B and D only charge one edge to the negative subcone and there are three negative cones (disregarding obstacles), there are at most three edge charged in negative subcones for each vertex. Taking polygonal obstacles into account and using the assumption that a vertex is part of at most one obstacle, there are at most four subcones in total and thus at most four edges charged by negative subcones to a vertex.
	
	\textbf{Positive subcones:} We shift our attention to the positive subcones and prove that each positive subcone, $C^u_j$, is charged for at most $2$ edges. As indicated in the charging scheme, only Scenarios A and C charge edges to positive subcones. 
	
	Scenario A charges at most $1$ charge to each positive subcone, since by construction, \graphname only adds one edge in each positive subcone of $v_i$ and $G_{10}$ does not add additional edges. 
	
	Scenario C charges at most $2$ edges to each positive subcone $C^{v_i}_j$, one for each of its neighbouring negative subcones $\overline{C^{v_i}_{j-1}}$ and $\overline{C^{v_i}_{j+1}}$. By the definition of the canonical path, every vertex $v_i$ along the path has at most one edge in $\overline{C^{v_i}_{j+1}}$ (and at most one in $\overline{C^{v_i}_{j-1}}$) that belongs to the canonical path of subcone $\overline{C^{u}_j}$. 
	
	The above two arguments show that every positive subcone has charge at most 3. Using a more careful analysis, we can show that each positive cone is charged by Scenario A or by Scenario C, but not by both. By definition, charges from Scenario A apply when $(u,v_i)$ is the shortest edge of $u$ in $\overline{C^{u}_j}$. Hence, if there exists a vertex $x$ in $\overline{C^{v_i}_{j-1}}$ or $\overline{C^{v_i}_{j+1}}$ such that $(x,v_i)$ is an edge in the canonical path, $v_i$ cannot be the closest vertex in $\overline{C^{u}_j}$ as $x$ lies above the horizontal line through $v_i$. Thus, $(u,v_i)$ would not be preserved during the construction of $G_{10}$. Therefore, if a positive subcone is charged for edge by Scenario A, it cannot have any edges charged by Scenario C. Conversely, if a positive subcone is charged by Scenario C, it cannot be charged with edges from Scenario A.
	
	Hence, a cone that is not affected by obstacles has at most $2$ charges. Next, we look at what changes when $v_i$ lies on an obstacle $P$. If $P$ does not split a cone, the above arguments still apply. Hence, we focus on the case where $P$ splits a positive cone $C^{v_i}_j$ into two subcones. Let $u$ be the closest vertex of $v_i$ in the right subcone of $C^{v_i}_0$ and let $u'$ be the closest vertex in its left subcone (see Figure~\ref{fig:G10_Degree_C}). 
	
	\begin{figure}[h!]
		\centering
		\includegraphics{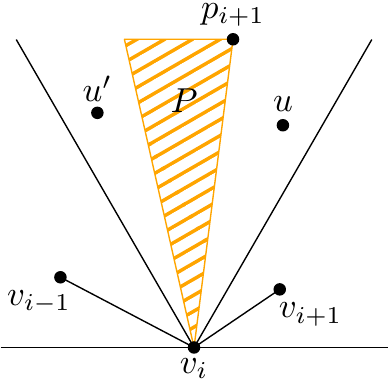}
		\caption{Two positive subcones of $v_i$ are charged for one edge by Scenario C.}
		\label{fig:G10_Degree_C}
	\end{figure}	
	
	By construction, the maximum number of charges from edges fitting Scenario A stays the same, $1$ per subcone. However, for edges from Scenario C, we show that the maximum charge per subcone reduces to $1$. Let $p_{i+1}$ be the vertex following $v_i$ on the boundary of $P$ on the side of $u$. Since $(u,v_i)$ is an edge of \graphname, $p_{i+1}$ is further from $v_i$ than $u$ is (or equal if $p_{i+1}=u$), as otherwise $p_{i+1}$ would be a closer visible vertex than $u$, contradicting the existence of $(u,v_i)$. For any vertex $v_{i-1}$ in $\overline{C^{v_i}_2}$ or $C^{v_i}_1$, $u$ is not visible since $(v_i,p_{i+1})$ blocks visibility and therefore $(v_{i-1},v_i)$ cannot be part of the canonical path in $\overline{C^u_0}$. We note that if $u = p_{i+1}$, $v_{i-1}$ may be visible to $u$, but in this case it lies in a different subcone of $u$ and hence would also not be part of this canonical path. Hence, $v_i$ as the leftmost vertex of the canonical path of $\overline{C^{u}_0}$, and using an analogous argument, $v_i$ is the rightmost vertex of the canonical path of $\overline{C^{u'}_0}$. Therefore, these positive subcones cannot be charged more than $1$ each by Scenario C. Using the fact that Scenario A and C cannot occur at the same time, we conclude that each positive cone is charged at most 2, regardless of whether an obstacle splits it.

	Thus, each positive cone is charged at most 2 and each negative subcone is charged at most 1. Since there are 3 positive cones and at most 4 negative subcones, the total degree bound is 10.
\end{proof}

Putting the results presented in this section together, we obtain the following theorem. 

\begin{theorem}
	$G_{10}$ is a plane $6$-spanner of the visibility graph of degree at most 10. 
\end{theorem}

\section{Shortcutting to Degree 7}
In order to reduce the degree bound from $10$ to $7$, we ensure that at most $1$ edge is charged to each subcone. According to Lemma~\ref{lem:G10Degree}, the negative subcones and the positive subcones created by splitting a cone into two already have at most $1$ charge. Hence, we only need to reduce the maximum charge of positive cones that are not split by an obstacle from $2$ to $1$. The two edges charged to a positive cone $C^v_i$ come from edges in the adjacent negative subcones $\overline{C^v_{i-1}}$ and $\overline{C^v_{i+1}}$ (Scenario C), whereas cone $C^v_i$ itself does not contain any edges. Hence, in order to reduce the degree from 10 to 7, we need to resolve this situation whenever it occurs in $G_{10}$. We do so by performing a transformation described below when a positive cone is charged twice for Scenario C. The graph resulting from applying this transformation to every applicable vertex is referred to as $G_7$. 

Let $u$ be a vertex in $G_{10}$ and let $v$ be a vertex on its canonical path (in $\overline{C^u_i}$) whose positive cone is charged twice. Let $(x,v)$ and $(v,y)$ be the edges charged to $\overline{C^u_i}$ and assume without loss of generality that $x$ occurs on the canonical path from $u$ to $y$ (see Figure~\ref{fig:G7}a). If neither $x$ nor $y$ is the closest vertex in the respective negative subcone of $v$ that contains them, we remove $(v,y)$ from the graph and add $(x,y)$ (see Figure~\ref{fig:G7}b). This reduces the degree of $v$, but may cause a problem at $x$. In order to solve this, we consider the neighbor $w$ of $x$ on the canonical path of $v$. Since $x$ is not the closest vertex in $v$'s subcone, this neighbor exists. If $w\in\overline{C^x_i}$ and $w$ is not the closest vertex of $x$ in the subcone that contains it, we remove $(x,w)$ (see Figure~\ref{fig:G7}c). 
	
	\begin{figure}[h!]
		\centering
		\subfloat[The situation before the transformation.]{
			\includegraphics{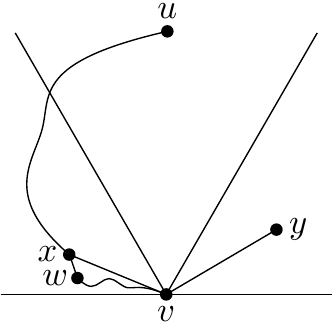}
		} \hspace{3em}
		\subfloat[After $(x,y)$ is added and $(v,y)$ is removed.]{
			\includegraphics{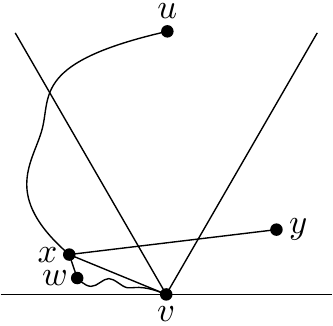}
		} \hspace{3em}
		\subfloat[The situation after the transformation.]{
			\includegraphics{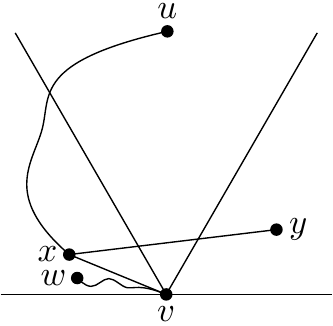}
		}
		\caption{Transforming $G_{10}$ into $G_7$.}
		\label{fig:G7}
	\end{figure}

The reasons for these choices are as follows. If an edge $(a,b)$ is part of a canonical path and $b$ is the closest vertex to $a$ in the negative subcone that contains it, then this edge is charged twice at $a$. It is charged once as the edge to the closest vertex (Scenario B) and once as part of the canonical path (Scenario C). Hence, we can reduce the charge to $a$ by 1 without changing the graph. 

In the case where we end up removing $(v,y)$, we need to ensure that there still exists a spanning path between $u$ and $y$ (and any vertex following $y$ on the canonical path). Therefore, we insert the edge $(x,y)$. This comes at the price that we now need to ensure that the total charge of $x$ does not increase. Fortunately, this is always possible. 

\begin{lemma}
	The degree of every vertex in $G_7$ is at most $7$.
\end{lemma}
\begin{proof}
	Since each vertex $v$ is part of at most one obstacle, $v$ can have at most $7$ subcones. Similar to $G_{10}$, we prove the degree bound of $G_7$ by demonstrating that the charge of each subcone after the transformation is at most $1$. Since we showed in Lemma~\ref{lem:G10Degree} that every negative subcone already has a charge of at most $1$ in $G_{10}$, we only need to prove that the transformation to construct $G_7$ reduces the maximum charge of each positive cone (not split by an obstacle) from $2$ to $1$. 
	
	Let $u$ be a vertex where $(u,v) \in G_\infty$ and $C^v_i$ is a positive cone that has $2$ edge charges in $G_{10}$. Let $x\in\overline{C^v_{i-1}}$ and $y\in\overline{C^v_{i+1}}$ where $(x,v)$ and $(v,y)$ are two edges in the canonical path of $\overline{C^{u}_i}$ (see Figure~\ref{fig:G7}a). If $x$ is the closest vertex of $v$ in $\overline{C^v_{i-1}}$, the edge $(x,v)$ is counted in both $\overline{C^v_{i-1}}$ (Scenario B) and $C^v_i$ (Scenario C). Then by simply removing the charge of $(x,v)$ in $C^v_i$ we reduce the charge of this positive cone to $1$ whilst maintaining that the total charge of a vertex is an upper bound on the degree of that vertex. Analogously, if $(v,y)$ is the closest vertex, it is counted twice at $v$. If neither $(x,v)$ nor $(v,y)$ is doubly counted and assuming, without loss of generality, that $x$ occurs on the canonical path from $u$ to $y$, we remove one of the two charges of $C^v_i$ by removing the edge $(v,y)$. Cone $C^v_i$ now has charge 1 as required. 
	
	If edge $(v,y)$ was removed, we then add the edge $(x,y)$ in the transformation. For vertex $y$, edge $(v,y)$ is removed and edge $(x,y)$ is added, thus the degree of $y$ remains unchanged. Charging $(x,y)$ to the subcone $C^y_{i+1}$ (the one that $(v,y)$ used to be charged to) will maintain the charges as well, ensuring that it is still an upper bound on the degree. 
	
	Next, we consider vertex $x$ whose degree increased when $(x,y)$ was added. Let $w$ be the neighbor of $x$ on the canonical path of $\overline{C^v_{i-1}}$. If $w$ is the closest vertex of $x$ in subcone $\overline{C^x_i}$, $(x,w)$ is counted twice: once in $C^x_{i-1}$ and once in $\overline{C^x_i}$. Removing the charge of $(x,w)$ from $C^x_{i-1}$ allows us to now charge $(x,y)$ to this cone, leading to the desired charges. If $(x,w)$ is not the closest vertex of $x$ in subcone $\overline{C^x_i}$, however, the last step of the transformation is performed and $(x,w)$ is removed. For vertex $x$, this maintains the degree, as edge $(x,y)$ is added and edge $(x,w)$ is removed. Edge $(x,y)$ is charged to the subcone that used to be charged for the removed edge $(u,w)$ and thus the charging bound of $\overline{C^x_i}$ remains unchanged. Finally, vertex $w$ only has an edge removed and thus the charge for $(x,w)$ is removed, resulting in its degree and charge to both be reduced by $1$. This completes the proof of the lemma.
\end{proof}

\begin{lemma}\label{lem:G7spanning}
	$G_7$ is a $3$-spanner of \graphname.
\end{lemma}
\begin{proof}
	For this lemma to stand, we have to prove that for any edge $(a,b)$ in \graphname, there exists a path from $a$ to $b$ in $G_7$ where the length of the path $\delta(a,b)\leq3\cdot|(a,b)|$. Since $G_{10}$ is proven to be a $3$-spanner of \graphname in Lemma~\ref{lem:G10Spanning}, it follows that there exists a spanning path between any pair of vertices in $G_7$ except between those affected during the removal of edges. Therefore, we focus only on the paths affected by the transformation. We will use the same notation as in the $G_7$ construction method. Hence, there are two edges potentially removed during the transformation from $G_{10}$ to $G_7$: $(v,y)$ and $(x,w)$. 
	
	The removal of $(v,y)$ interrupts the path between $x$ and $y$ along the canonical path of $\overline{C^{u}_0}$ and hence the paths from $u$ to any vertices on the canonical path after $y$. The original spanning path in $G_{10}$ includes $(x,v)$ and $(v,y)$. When removing $(v,y)$, we also add edge $(x,y)$, so as the new spanning path in $G_7$, we consider the same path as that in $G_{10}$ where $(x,v)$ and $(v,y)$ are replaced by $(x,y)$. By triangle inequality, $|xy|\leq|xv|+|vy|$, therefore the spanning ratio of the spanning path from $u$ to $y$ and any vertices after $y$ is still at most $3$. 
	
	Removing edge $(v,y)$ also affects the spanning ratio between $v$ and $y$. By construction, however, $(v,y)$ is only removed if $y$ is not the closest visible vertex of $v$ in $\overline{C^v_1}$. Hence, there exists a canonical path between $v$ and $y$ in $\overline{C^v_1}$. According to Lemma~\ref{lem:G10Spanning}, the length of this canonical path is at most $3\cdot|vy|$ and by the preceding paragraph this spanning ratio is maintained. Thus, there still exists a spanning path between $v$ and $y$. 
	
	Next, we consider the effect of removing edge $(x,w)$. By construction, if it is removed, $w$ is not the closest vertex to $x$ in $\overline{C^{x}_0}$. We first show by contradiction that $x$ is the end of the canonical path of $\overline{C^v_2}$ and thus has only one neighbour, $w$, on this path. Let vertex $z$ be the other neighbor of $x$ along the canonical path of $\overline{C^v_2}$. Consider the edge $(z,v) \in G_\infty$. This edge either intersects the boundary of $\triangle uxv$ or $z$ lies in its interior. However, by Lemma~\ref{lem:canonpathempty}, $z$ cannot lie inside and by Lemma~\ref{lem:lem3} $G_\infty$ is plane. Hence, $z$ cannot exist and thus $x$ is the last vertex on the canonical path of $\overline{C^v_2}$. Since there are no further edges on the canonical path of $\overline{C^v_2}$, when $(x,w)$ is removed, the only paths affected are the one between $x$ and $w$ and the canonical path from $v$ to $x$. 
	
	Since $w$ is not the closest vertex to $x$ in $\overline{C^x_0}$, there exists a canonical path of $\overline{C^{x}_0}$ from $x$ to $w$. As stated in Lemma~\ref{lem:G10Spanning}, this path is a $3$-spanning path of $(x,w)$. Therefore, there exists a spanning path between $x$ and $w$. 
		
	Once $(x,w)$ is removed, the path from $x$ to $v$ via the canonical path of $\overline{C^v_2}$ no longer exists. However, since the edge $(x,v)$ is not removed in the transformation, there still exists a 1-spanning path between $x$ and $v$. 
	
	Since for all affected paths, there is still a spanning path with length at most $3$ times the length of the original path in $G_\infty$, $G_7$ is a 3-spanner of \graphname.
\end{proof}

It remains to show that $G_7$ is a plane graph.

\begin{lemma}
	$G_7$ is a plane graph.
\end{lemma}
\begin{proof}
	By construction, each transformation performed when converting $G_{10}$ to $G_7$ includes the addition of at most one edge $(x,y)$ and the removal of at most two edges $(v,y)$ and $(x,w)$. By Lemma~\ref{lem:G10Plane}, $G_{10}$ is a plane graph and thus the removal of $(v,y)$ and $(x,w)$ does not affect the planarity of $G_7$. It remains to consider the added edges. Specifically, we need to show that: $(x,y)$ cannot intersect any edge from $G_{10}$, $(x,y)$ cannot intersect an obstacle, and $(x,y)$ cannot intersect another edge added during the construction of $G_7$. 
	
	We first prove that the edge $(x,y)$ cannot intersect any edge from $G_{10}$. Since $x$, $v$, and $y$ are consecutive vertices in the canonical sequence of a negative subcone of $u$,  according to Lemma~\ref{lem:canonpathempty}, $\triangle uxv$ and $\triangle uvy$ are empty. Furthermore, since the boundary edges of these triangles are edges in $G_\infty$, which is plane by Lemma~\ref{lem:lem3}, no edges or obstacles can intersect these edges. Hence, since $(x,y)$ lies inside the quadrilateral $uxvy$, edge $(u,v)$ is the only edge that can intersect $(x,y)$. However, since $x$ and $y$ lie in $\overline{C^v_{i-1}}$ and $\overline{C^v_{i+1}}$, both $x$ and $y$ are closer to $u$ than $v$ in $\overline{C^{u}_i}$, implying that $(u,v)$ is not part of $G_{10}$. Therefore, $(x,y)$ does not intersect any edge from $G_{10}$. 
	
	Next, we show that $(x,y)$ cannot intersect any obstacles. Since by Lemma~\ref{lem:canonpathempty} no obstacles, vertices, or edges can pass through or exist within $\triangle uxv$ and $\triangle uvy$, the only obstacles we need to consider are $\triangle uxv$ and $\triangle uvy$ themselves and the line segment obstacle $uv$. We note that in all three cases the obstacle splits the negative cone $\overline{C^{u}_i}$ into two subcones, and $x$ and $y$ lie in different negative subcones and thus on different canonical paths (see Figure~\ref{fig:G7_plane}). This implies that $C^u_i$ is not charged twice by the same canonical path and thus, the transformation to $G_7$ would not add $(x,y)$. Therefore, $(x,y)$ cannot intersect any obstacles. 
	
	\begin{figure}[h!]
		\centering
		\includegraphics{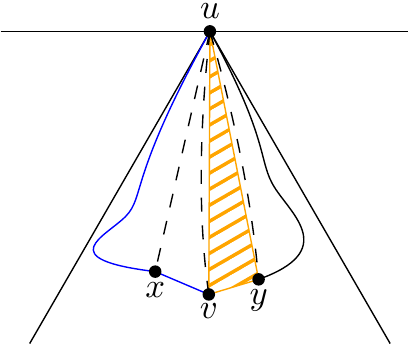}
		\caption{The quadrilateral $uxvy$ in $G_7$ where $\triangle uvy$ is an obstacle itself. $\overline{C^u_0}$ is split into two subcones and $x$ and $y$ lie on two different canonical paths.}
		\label{fig:G7_plane}
	\end{figure}
	
	Finally, we show that the added edge $(x,y)$ cannot intersect another edge added during the transformation from $G_{10}$ to $G_7$. We prove this by contradiction. Let $e$ be an added edge that intersects $(x,y)$. Edge $e$ cannot have $x$ or $y$ as an endpoint, as this would imply that it does not intersect $(x,y)$. Furthermore, since by Lemma~\ref{lem:canonpathempty} $\triangle uxv$ and $\triangle uvy$ are empty and $e$ cannot intersect edges $(x,v)$ and $(v,y)$, as shown above, $v$ is one of the endpoints of $e$. By construction of $G_7$, $e$ is added because either $x$ or $y$ has a positive cone charged by two edges from its adjacent negative subcones (one of which being the edge to $v$). However, this contradicts the fact that $v$ lies in their positive cones ($C^x_{i-1}$ and $C^y_{i+1}$) implied by the edge $(x,y)$ being added during the transformation. Therefore, $(x,y)$ cannot intersect any edge added while constructing $G_7$, completing the proof. 
\end{proof}

We summarize our main result in the following theorem. 

\begin{theorem}
	$G_7$ is a plane $6$-spanner of the visibility graph of degree at most 7. 
\end{theorem}

\section{Conclusion}
Our goal was to expand upon the constrained half-$\Theta_6$-graph suggested by Bose~\etal~\cite{basepaper} supporting more complex obstacles in the form of convex polygons, as opposed to line segments obstacles, whilst bounding the degrees of the resulting spanners. We presented a plane spanning graph \graphname of the visibility graph with spanning ratio $2$ that avoids polygonal obstacles. In addition, we presented three bounded-degree 6-spanners. $G_{15}$ is constructed by removing edges in the negative subcones and results in a graph of degree 15. When we allow vertices to require the presence of edges between their neighbors, we can build $G_{10}$, reducing the degree bound to $10$. Finally, if we are also allowed to add edges that are not part of $G_\infty$, we showed how to construct $G_7$, a plane $6$-spanner of the visibility graph of degree $7$.

With the construction of these spanners, future work includes developing routing algorithms for them. There has been extensive research into routing algorithms, including routing on the visibility graph in the presence of line segment obstacles~\cite{BFRV2017RoutingJournal,BKRV2017Routing,BKRV2018RoutingVisibilityGraphJournal}. For polygonal obstacles, Banyassady~\etal~\cite{BKMRRSSVW2017Routing} developed a routing algorithm on the visibility graph with polygonal obstacles, though this work requires some additional information to be stored at the vertices. 

Another logical next step is to reduce the degree bound further. Without obstacles, Bonichon~\etal~\cite{Bonichon2015} constructed a plane 156.82-spanner of degree 4. Kanj~\etal~\cite{DBLP:journals/corr/KanjPT16} improved on this by constructing a plane $20$-spanner with maximum degree $4$. Extending this to the setting with obstacles would be an interesting improvement on this paper. Additionally, Biniaz~\etal~\cite{DBLP:journals/corr/BiniazBCGMS16} showed how to construct plane spanners of degree $3$ in two special cases. 

Finally, it would be interesting to see if different constructions can provide bounded-degree spanners with smaller spanning ratios, even at the cost of increasing the degree bound slightly. Recently, Bose~\etal~\cite{Bose2018} showed how to construct an (approximately) $4.414$-spanner of degree 8, though not in the presence of obstacles.

\bibliographystyle{plainurl}
\bibliography{references.bib}

\end{document}